\documentclass[a4paper,cleveref,autoref,thm-restate]{lipics-v2021}

\nolinenumbers

\title{Lower Bounds on Retroactive Data Structures}

\author{Lily Chung}{Massachusetts Institute of Technology, Cambridge, MA, USA}{ikdc@mit.edu}{https://orcid.org/0000-0001-7056-6155}{}
\author{Erik D. Demaine}{Massachusetts Institute of Technology, Cambridge, MA, USA}{edemaine@mit.edu}{https://orcid.org/0000-0003-3803-5703}{}
\author{Dylan Hendrickson}{Massachusetts Institute of Technology, Cambridge, MA, USA}{dylanhen@mit.edu}{https://orcid.org/0000-0002-9967-8799}{}
\author{Jayson Lynch}{Cheriton School of Computer Science, University of Waterloo,  Waterloo, ON, Canada}{jaysonl.lynch@uwaterloo.ca}{}{}

\authorrunning{L. Chung, E.\,D. Demaine, D. Hendrickson, and J. Lynch}

\Copyright{Lily Chung, Erik D. Demaine, Dylan Hendrickson, and Jayson Lynch}

\ccsdesc{Theory of computation~Design and analysis of algorithms~Data structures design and analysis~Cell probe models and lower bounds}
\keywords{Retroactivity, time travel, rollback, fine-grained complexity}

\acknowledgements{%
This work was initiated during open problem solving in the MIT class on
Advanced Data Structures (6.851) in Spring 2021.
We thank the other participants of that class --- in particular,
Joshua Ani, Josh Brunner, and Naveen Venkat ---
for related discussions and providing an inspiring atmosphere.
We also thank Michael Coulombe for helpful pointers regarding time travel.
}

\newcommand\metaop{\textsf}
\newcommand\partially[1]{\metaop{Partial-Retro}(#1)}
\newcommand\fully[1]{\metaop{Full-Retro}(#1)}
\newcommand\lazy[1]{\metaop{Lazy}(#1)}
\newcommand\lazyop{\metaop{Lazy}}
\newcommand\op{\textsc}
\newcommand\evaluate{\op{evaluate}}
\newcommand\domain{\operatorname{domain}}
\newcommand\codomain{\operatorname{codomain}}
\newcommand\Insert[1]{\op{Insert}_{#1}}
\newcommand\Query[1]{\op{Query}_{#1}}
\newcommand\Delete{\op{Delete}}
\newcommand\ADT{\mathcal{A}}
\newcommand\DS{\mathcal{D}}

\theoremstyle{definition}
\newtheorem{definitionx}[theorem]{Definition}

\newtheorem{conjecturex}{Conjecture}

\let\epsilon=\varepsilon
\def\defn#1{\textbf{\textit{\boldmath #1}}}

\EventEditors{Sang Won Bae and Heejin Park}
\EventNoEds{2}
\EventLongTitle{33rd International Symposium on Algorithms and Computation (ISAAC 2022)}
\EventShortTitle{ISAAC 2022}
\EventAcronym{ISAAC}
\EventYear{2022}
\EventDate{December 19--21, 2022}
\EventLocation{Seoul, Korea}
\EventLogo{}
\SeriesVolume{248}
\ArticleNo{32}

\begin{document}

\maketitle

\begin{abstract}
We prove essentially optimal fine-grained lower bounds on the gap between a data structure and a partially retroactive version of the same data structure. Precisely, assuming any one of three standard conjectures, we describe a problem that has a data structure where operations run in $O(T(n,m))$ time per operation, but any partially retroactive version of that data structure requires $T(n,m) \cdot m^{1-o(1)}$ worst-case time per operation, where $n$ is the size of the data structure at any time and $m$ is the number of operations. Any data structure with operations running in $O(T(n,m))$ time per operation can be converted (via the ``rollback method'') into a partially retroactive data structure running in $O(T(n,m) \cdot m)$ time per operation, so our lower bound is tight up to an $m^{o(1)}$ factor common in fine-grained complexity.
\end{abstract}

\section{Introduction}

A trope in popular science fiction is the idea of traveling back in time,
making a change to the world,
then traveling forward in time (usually to the original time)
to observe the effected changes.
This style of time travel (as opposed to stable time loops or multiverses)
was one of the inspirations for retroactive data structures.
Such sequences of events occur in many pieces of media including in the movies
\emph{Back To The Future} (1985),
\emph{Timecop} (1994),
\emph{12 Monkeys} (1995),
\emph{Star Trek: First Contact} (1996),
\emph{The Butterfly Effect} (2004),
\emph{Looper} (2012), and
\emph{Avengers: Endgame} (2019);
TV shows
\emph{Doctor Who} (1963/2005--),
\emph{Heroes} (2006--),
\emph{Marvel's Agents of S.H.I.E.L.D.} (2013--),
\emph{DC's Legends of Tomorrow} (2016--), and
\emph{The Umbrella Academy} (2019--);
anime/manga \emph{Doraemon} (1969/1973--) and \emph{Steins;Gate} (2009/2011--);
short story \emph{A Sound of Thunder} (Ray Bradbury, 1952);
books \emph{The Hitchhiker's Guide to the Galaxy} (Douglas Adams, 1979) and
\emph{Pastwatch: The Redemption of Christopher Columbus}
(Orson Scott Card, 1996);
and video games \emph{Chrono Trigger} (1995) and
\emph{The Legend of Zelda: Ocarina of Time} (1998).

Is this kind of ``jump back, change, jump forward'' or ``Back To The Future''
time travel realistic?
It is inconsistent with known models of physics, but it is nonetheless
interesting to prove inconsistency with other assumptions about the
physical world.
Here we give evidence of such inconsistency under a
\defn{fine-grained physical Church--Turing thesis}:
the universe's physics while traveling forward in time
are implemented by a computer with similar power
(up to polylogarithmic or $n^{o(1)}$ factors)
to standard theoretical computers such as the word RAM,
possibly randomized or quantum,
or possibly a parallel system of such computers.
(The word RAM is the standard in data structures and most of theoretical
computer science, and it can efficiently simulate
e.g.\ an $O(1)$-head $d$-dimensional Turing machine for any~$d$,
making it a good choice of universal model, at least for lower bounds.)
Under this assumption, is time travel \emph{computationally} realistic?

We can turn such questions about computational realism into problems about
data structures, by thinking of the state of the universe as a data structure.
For example, jumping back in time (for viewing only) is essentially asking
for a \emph{partially persistent} universe computation,
while jumping back in time and making a change to branch a new universe
is essentially asking for a \emph{fully persistent} universe computation,
both of which are efficiently possible for essentially any universe computation
(with logarithmic overhead in general, and constant overhead assuming
bounded in-degree which may be reasonable given geometric constraints)
\cite{persistence}.
Returning to ``Back to the Future'' time travel,
making a (blind) change in the past and instantly seeing the cascaded effects
on the present is essentially asking for a \emph{partially retroactive}
universe computation, while observing and changing the past and immediately
updating the future is essentially asking for a \emph{fully retroactive}
universe computation.%
\footnote{For this reason, an early unused name for retroactive data structures
  was ``time-travel data structures'', as evidenced by the keyword list of
  \cite{demaine2007retroactive}.}

\subsection{Our Results}
\label{sec:results}

In this paper, we give new evidence that data structures cannot be made
efficiently retroactive, even partially retroactive,
which complements a prior separation between partial and full retroactivity
\cite{chen2018nearly}.  Specifically,
under any of three standard assumptions in fine-grained complexity,
we prove near-optimality of the trivial ``rollback'' method
for retroactive changes to the past \cite[Theorem~1]{demaine2007retroactive}:
rewind time by undoing operations to when the change should occur,
make the desired change, and then replay the previously undone operations
(by maintaining the current timeline of changes and their inverses).
In the worst case (and assuming no amortization in the input data structure),
this method occurs a multiplicative overhead of $\Theta(m)$
where $m$ is the number of operations in the timeline.
We prove an essentially matching worst-case lower bound of $\Omega(m^{1-o(1)})$
multiplicative overhead under any one of the following assumptions,
which are all common in fine-grained complexity
\cite{williams2018some,chen2018nearly} and were the same assumptions
made by the prior retroactive separation results \cite{chen2018nearly}:
\begin{enumerate}
\item
  CircuitSAT on an $n$-input $2^{o(n)}$-gate circuit
  requires $2^{n-o(n)}$ time.

  This assumption is weaker than the standard
  Strong Exponential Time Hypothesis (SETH) \cite{impagliazzo2001problems}
  which makes a similar statement about formulas
  (which cannot re-use computations like circuits can).
\item
  $(\min,+)$ matrix--vector product with an $n \times n$ integer matrix and
  an online sequence of $n$ integer vectors of length $n$
  requires $n^{3-o(1)}$ time.

  This assumption is weaker than the standard APSP assumption
  that All-Pairs Shortest Paths in an $n$-vertex graph with integer weights
  (equivalent to \emph{offline} $(\min,+)$ matrix--vector product)
  requires $n^{3-o(1)}$ time \cite{henzinger2015unifying}.
\item
  3SUM (given $n$ integers, do any three sum to zero?)\
  requires $n^{2-o(1)}$ time \cite{gajentaan1995class}.

  This assumption is arguably the beginning of fine-grained complexity.
\end{enumerate}
Our lower bounds hold even if the retroactive data structure is amortized,
and apply to randomized algorithms if the assumptions do.

Applied to time travel, our results show that making a change in the past
and then jumping forward to the present should take roughly as much computation
as the universe needed to originally advance through that much time.
In other words, adding ``Back To The Future'' time travel to the universe
requires the universe computation to slow down by a factor nearly linear
in the lifetime of the universe, which seems computationally unrealistic.
Here we assume the fine-grained physical Church--Turing Thesis described above,
and ignore possible speedups from quantum computation
(which remains an area for future research).
If the universe is a parallel system of computers, then
our worst-case instance consists of a proportional number of instances of
our worst-case data structures.
Our result can be seen as justification for the model of time travel
implemented in the video game \emph{Achron} (2011), where changing the past
propagates a wave of consequential changes to the future, but the wave
moves ``slowly'' (a constant factor faster than normal forward time travel).

\subsection{Related Work}

The original paper on retroactive data structures \cite{demaine2007retroactive}
proves an $\Omega(m)$ worst-case lower bound on the multiplicative overhead
required for partial retroactivity, where $m$ is the number of operations in the timeline \cite[Theorem~2]{demaine2007retroactive}.
Asymptotically, this result is an improvement over our result,
matching the $O(m)$ upper bound up to a constant factor
instead of an $m^{o(1)}$ factor.
The difference is in the model: the existing $\Omega(m)$ lower bound holds in
the history-dependent algebraic-computation-tree model,
the integer RAM, and generalized real RAM.
(The reduction is from online polynomial evaluation,
which provably requires linear time in these models.)
But the result is not known to hold on the word RAM, for example,
where the best known unconditional lower bound on retroactivity
is $\Omega(\sqrt{m/\log m})$ \cite[Theorem~3]{demaine2007retroactive}.

By contrast, our results hold wherever the fine-grained assumptions above hold,
which so far seems to include all ``reasonable'' models of computation.
For example, the fastest known algorithm for 3SUM on a word RAM
runs in $O\left(n^2 \big(\frac{\lg \lg n}{\lg n}\big)^2\right)$ time
\cite{Baran-Demaine-Patrascu-2007}.
This bound is roughly a logarithmic factor better than
the fastest known algorithm for 3SUM on a real RAM, which runs in
$O\big(n^2 \, \frac{\lg \lg n}{\lg n}\big)$ time \cite{Freund-2017},
an improvement on the recent subquadratic breakthrough
\cite{Gronlund-Pettie-2018}.
But all of these algorithms are quadratic up to polylogarithmic factors.

On the other hand, in the decision tree model of computation,
3SUM can be solved in $O(n^{3/2} \sqrt{\log n})$ time
\cite{Gronlund-Pettie-2018}, breaking the 3SUM conjecture.
However, this model of computation is generally considered \emph{unreasonable}
because it allows the choice of algorithm to depend on the problem size~$n$.
Further, like most such decision-tree results, the proof is not ``algorithmic'':
the best known algorithm to compute which algorithm to run for a given $n$
uses exponential time. \looseness=-1

Our work is closely based on another paper proving a separation between
partial and full retroactivity \cite{chen2018nearly}.
In particular, our fine-grained complexity assumptions are the same as theirs.
Under these assumptions, Chen et al.~\cite{chen2018nearly} proved that
the worst-case separation between a partially retroactive and a
fully retroactive data structure for the same underlying problem
is a multiplicative factor of $\Omega(m^{1/2-o(1)})$,
which is essentially tight against the known upper bound of $O(\sqrt m)$
\cite{demaine2007retroactive}.

This paper is not the first to consider applications
of computational complexity to time travel.
Aaronson, Bavarian, and Giusteri \cite{Aaronson-Bavarian-Giusteri-2016}
show that consistent timelines from closed timelike curves
(another common model for time travel in popular science fiction)
require solving PSPACE-complete problems.
This result suggests that that this model of time travel,
despite being plausible in existing models of physics,
is not actually computationally feasible.
Our result can be viewed as a complementary negative result
for the ``Back To The Future'' model of time travel.
(However, our result shows only a conjectured linear separation,
whereas \cite{Aaronson-Bavarian-Giusteri-2016} shows a conjectured
exponential separation.)
Both negative results only apply to ``long-distance'' time travel;
making changes in the recent past remains computationally feasible,
leaving ample room for science fiction.
Both results also leave open the possibility that a time traveler
jumping forward simply does not experience time
(their subjective time is frozen) while the universe rolls time forward.

\subsection{Techniques}
\label{sec:techniques}

Our approach builds on previous work by Chen et al.~\cite{chen2018nearly},
which proves an $\tilde \Omega(\sqrt m)$ gap between partial and full
retroactivity under the same fine-grained complexity assumptions,
where $m$ is the number of operations in the data structure's timeline.
This result implies an $\tilde \Omega(\sqrt m)$ worst-case separation
between a basic (nonretroactive) data structure and a fully retroactive
version of that data structure.
Our work improves on this separation in two ways:
it provides a separation between a basic data structure and a
\emph{partially} retroactive version of that data structure,
and it improves the separation from $\tilde \Omega(\sqrt m)$ to
$\tilde \Omega(m)$.

To achieve these improvements, we adapt Chen et al.'s constructions
using the following techniques:
\begin{enumerate}
\item We define the ``lazy'' version of an arbitrary abstract data type (data structure interface), which allows us to move all queries to the end while still preserving their answers, making partial and full retroactivity essentially equivalent. We prove a general result (Theorem~\ref{thm:lazy conversion}) which lets us translate gaps for full retroactivity into gaps for partial retroactivity.
\item We use fewer operations than Chen et al.\ by combining a long sequence of simple operations into one more complicated operation. In particular, Chen et al.'s data structures use one operation to set a single element of a list, while ours set the entire list in a single operation. At the cost of operations being more expensive, this reduces the number of operations performed. As a result, the same gap is larger when considered as a function of the number of operations.
\end{enumerate}
Using Technique~1 alone, applying Theorem~\ref{thm:lazy conversion} to Chen et al.'s results, we obtain an $\tilde \Omega(\sqrt m)$ gap between nonretroactivity and partial retroactivity, which is interesting but not optimal. Using Technique~2 without Technique~1 does not achieve anything new, because Chen et al.\ use a large number of nonconsecutive query operations. To significantly reduce the number of operations, we need to batch the query operations into a single more complicated operation, which our lazy transform makes possible by moving them all to the end.

\subsection{Outline}

In Section~\ref{sec:defs}, we provide the definitions needed for this paper, including fully and partially retroactive data structures, and the three fine-grained complexity assumptions.

In Section~\ref{sec:lazy}, we define lazy data structures, and use them to prove a result that converts gaps for full retroactivity into gaps for partial retroactivity. Applying this to Chen et al.'s work gives a $\Omega(\sqrt m)$ gap between a base data structure and the partial retroactive version, under each of the three fine-grained complexity assumptions.

In the remaining sections, we prove a stronger $\Omega(m)$ slowdown for partial retroactivity under each of the same assumptions. Each proof is based on a construction due to Chen et al.~\cite{chen2018nearly}, with $\lazyop$ applied.
We also apply the second idea described above--merging consecutive operations into a single more complicated operation--which gives the improvement from $\Omega(\sqrt m)$ to $\Omega(m)$.
We are also able to make a few simplifications: first, we do not always need the full power of inserting queries and calling \evaluate, and can use an ADT which can only query some simpler summary of past states. Second, Chen et al.'s lower bound from CircuitSAT is more complicated than necessary.

We prove this $\Omega(m)$ separation under Conjecture~\ref{conj:circuitsat} about CircuitSAT in Section~\ref{sec:circuitsat}, under Conjecture~\ref{conj:online min + product} about Online $(\min,+)$ Product in Section~\ref{sec:online min + product}, and finally under Conjecture~\ref{conj:3sum} about 3SUM in Section~\ref{sec:3sum}.

\section{Definitions}
\label{sec:defs}

In this section, we introduce a simple formalism for data structures and
their problem/interface specifications (ADTs) that lets us to define
general transformations on those specifications.
This enables formal definitions of ``a retroactive version of a
data structure'' as well as our {\lazyop} transformation,
which in turn help us state our results precisely.

\begin{definitionx}[ADT]
	\label{def:ADT}
	An \defn{abstract data type} (\defn{ADT}) $\ADT$ consists of
	\begin{itemize}
		\item A set $\mathcal U$ of \defn{update} operations, where each $u \in \mathcal U$ takes an argument in $\domain(u)$ but does not return a value. (Instead, each update influences the results of future queries.)
		\item A set $\mathcal Q$ of \defn{query} operations, which each $q \in \mathcal Q$ takes an argument in $\domain(q)$ and returns a value in $\codomain(q)$. (Queries cannot affect the result of future queries.)
		\item A partial function $\mathcal A$ specifying the result of a query operation given the sequence of prior calls to update operations. That is, given a sequence $\hat u=[u_1(x_1), \dots, u_k(x_k)]$ of calls to update operations and a call to a query operation $q(x)$, the ADT may specify the result $\ADT(\hat u, q(x))$ of $q(x)$ after calling precisely the operations in $\hat u$. Note that the result of $q(x)$ can depend on prior update operations, but not on prior query operations. Some behavior may be left \defn{undefined} (meaning that any result is valid).%
    \footnote{For convenience in future ADT definitions, we allow a data structure to output anything when behavior is not defined, but in all of our results it is easy to test for undefined behavior, and thus one could modify the ADTs we define to have no undefined behavior by defining a ``default'' return value, without affecting our results.}
	\end{itemize}
\end{definitionx}

For instance, the \textsc{Dictionary} ADT has two update operations \op{insert} and \op{delete}, and one query operation \op{contains}, all of which take an integer as an argument. The desired behavior is that \op{contains}$(x)$ returns True if \op{insert}$(x)$ has been called more recently than $\op{delete}(x)$, and False otherwise. This ADT represents a set which starts empty, and can have elements inserted and deleted.

\begin{definitionx}[Data structure, implementation]
	A \defn{data structure} consists of some internal storage (e.g., a pointer machine or word RAM), and some algorithms acting on this internal storage, starting from a specified initial state. A data structure \defn{implements} an ADT $\ADT$ if the algorithms are labelled with the operations of $\ADT$, and have the correct behavior:
	suppose we run a sequence of operation calls on the data structure, ending in a query $q(x)$. Let $\hat u$ be the sequence of these operations that are updates. Then the final query must return $\ADT(\hat u, q(x))$, if this value is defined.%
\end{definitionx}

An ADT may have many different implementations; for example, the \textsc{Dictionary} ADT can be implemented by hash tables or balanced search trees. The usual goal of data structure design is to invent faster implementations for a given ADT.

\subsection{Retroactivity}

Now we can define partial and full retroactivity as transformations on ADTs.

\begin{definitionx}[Partially retroactive]
	Let $\ADT$ be an ADT. We define an ADT called \defn{partially retroactive $\ADT$}, written $\partially{\ADT}$. It has the following operations:
	\begin{itemize}
		\item For each update operation $u$ of $\ADT$, an update operation $\Insert{u}$ with domain $\mathbb N \times \domain(u)$.%
		\item An update operation $\Delete$ with domain $\mathbb N$.
		\item The same query operations as $\ADT$.
	\end{itemize}
	To describe the expected behavior, we imagine that a data structure maintains a sequence $\hat u$ of calls to update operations of $\ADT$ (the ``timeline''). Then $\Insert{u}(k,x)$ means we insert $u(x)$ at position $k$ in this list, and $\Delete(k)$ means we remove the call at position $k$. A call to query operation $q(x)$ in $\partially{\ADT}$ should return $\ADT(\hat u, q(x))$ if that value is defined; that is, we evaluate $q(x)$ at the end of the sequence of updates, which is thought of as the ``present''.
\end{definitionx}

\begin{definitionx}[Fully retroactive]
	Let $\ADT$ be an ADT. We define an ADT called \defn{fully retroactive $\ADT$}, written $\fully{\ADT}$. It has the following operations:
	\begin{itemize}
		\item For each update operation $u$ of $\ADT$, an update operation $\Insert{u}$ with domain $\mathbb N \times \domain(u)$.%
		\item An update operation $\Delete$ with domain $\mathbb N$.
		\item For each query operation $q$ of $\ADT$, a query operation $\Query{q}$ with domain $\mathbb N\times \domain(q)$ and the same codomain as~$q$.
	\end{itemize}
	As with partially retroactive $\ADT$, we imagine that a data structure maintains a sequence $\hat u$ of calls to update operations of $\ADT$. 
	The meanings of $\Insert{u}$ and $\Delete$ are the same as above. 
	A call $\Query{q}(k,x)$ in $\fully{\ADT}$ should return $\ADT(\hat u_{1\cdots k}, q(x))$ if that value is defined, where $\hat u_{1\cdots k}$ is the first $k$ elements of $\hat u$; that is, it should give the result of calling $q(x)$ after time $k$ in the timeline.
\end{definitionx}

\begin{lemma}[{\cite[Theorem~1]{demaine2007retroactive}}]
  \label{lem:rollback}
  Suppose $\ADT$ has an implementation in which each operation takes $O(T)$ time.
  Then $\partially{\ADT}$ and $\fully{\ADT}$ have implementations
  in which each operation takes $O(m \, T)$ time,
  where $m$ denotes the number of calls to update operations in the timeline.
  Furthermore, the $\partially{\ADT}$ implementation can support
  query operations in the same time bound as~$\ADT$.
\end{lemma}

\begin{proof}
Both partial and full retroactivity can be achieved by the simple rollback method in which all operations and the changes they make to the data structure's internal storage get recorded in a stack. These operations can then be ``rolled back'', undoing their effect; the new update operation applied at the appropriate location; and then all rolled back operations re-applied.
\end{proof}

\subsection{Complexity Assumptions}
\label{sec:Complexity Assumptions}

Our results rely on the same computational complexity assumptions
used by Chen et al.~\cite{chen2018nearly}
to show their gap between partial and full retroactivity.
These assumptions are implied by (i.e. are weaker than) the three main conjectures
used in fine-grained complexity --- SETH, the APSP Conjecture, and
the 3SUM Conjecture~\cite{williams2018some,chen2018nearly} --- respectively. 

\begin{conjecturex}\label{conj:circuitsat}
	Time $2^{n-o(n)}$ is required to solve \defn{SIZE$(2^{o(n)})$ CircuitSAT}: given an $n$-input circuit of size $2^{o(n)}$, decide whether it is satisfiable.
\end{conjecturex}

Boolean circuits are more general structures than Boolean formulas, making this conjecture weaker, and thus more believable than, the foundational Strong Exponential Time Hypothesis (SETH)~\cite{impagliazzo2001problems}.

\begin{conjecturex}\label{conj:online min + product}
	Time $n^{3-o(1)}$ is required to solve \defn{Online $(\min,+)$ Product}: given an $n\times n$ integer matrix $A$, and $n$ vectors $v_1,\dots,v_n$ which are revealed one by one, compute each $(\min,+)$ product $A\diamond v_i$. 
	The next vector $v_{i+1}$ is revealed after outputting $A\diamond v_i$. 
	The $(\min,+)$ product $A\diamond v$ is an $n$-component vector whose $j$th component is defined as
	$$(A\diamond v)_j = \min_{k=1}^n(A_{j,k}+v_k).$$
\end{conjecturex}

The offline (and thus easier) version of Online $(\min,+)$ Product
\cite{henzinger2015unifying} is equivalent to the well-known all-pairs shortest path problem, which is commonly assumed hard in fine-grained complexity. Online $(\min,+)$ Product is also a generalization of Online Boolean Matrix--Vector Product, another problem commonly used in fine-grained complexity and not currently known to be equivalent to APSP~\cite{henzinger2015unifying}.

\begin{conjecturex}\label{conj:3sum}
	Time $n^{2-o(1)}$ is required to solve \defn{3SUM}:
  given three size-$n$ sets $A$, $B$, and $C$ of integers,
  decide whether there exists $(a,b,c)\in A\times B\times C$ such that $a+b+c=0$.
\end{conjecturex}

The 3SUM conjecture was one of the first and remains one of the main conjectures
in fine-grained complexity~\cite{gajentaan1995class}.
The version of 3SUM stated in Conjecture~\ref{conj:3sum} is actually called
3SUM$'$ in \cite{gajentaan1995class},
while Section~\ref{sec:results} states the original 3SUM problem.
But there is a simple linear-time reduction between 3SUM and 3SUM$'$
\cite[Theorem~3.1]{gajentaan1995class},
so these two versions of the 3SUM conjecture are equivalent.
We use the version stated in Conjecture~\ref{conj:3sum} in our construction.

\section{Transforming Partially Retroactive Transformations into Fully Retroactive Transformations}\label{sec:lazy}

In this section, we develop the first technique mentioned in Section~\ref{sec:techniques}. 
First we define a ``lazy'' version of an abstract data type (Definition~\ref{def:ADT}) $\ADT$. 
Roughly speaking, $\lazy{\ADT}$ makes the queries of $\ADT$ record their value but not return anything, and adds a new query operation $\evaluate$ to extract the recorded result of a query. This transformation allows us to simulate full retroactivity using partial retroactivity: insert a ``query'' that is actually an update operation at some point in the past, and then call $\evaluate$ in the present to read its result. This is the idea behind the main result of this section, Theorem~\ref{thm:lazy conversion}.

\begin{definitionx}[Lazy]
	Let $\ADT$ be an ADT. We define an ADT called \defn{lazy $\ADT$}, written $\lazy{\ADT}$. It has the following operations:
	\begin{itemize}
		\item The same update operations as $\ADT$.
		\item For each query operation $q$ of $\ADT$, an \emph{update} operation $q^\dagger$ with the same domain as~$q$.
		\item A query operation $\evaluate$ with domain $\mathbb N$ and codomain $\displaystyle \bigsqcup_{\text{query }q\text{ of }\ADT} \codomain(q)$.%
	\end{itemize}
	Suppose $\hat u^\dagger$ is the sequence of calls to update operations of $\lazy{\ADT}$, and suppose the $k$th call $u^\dagger_k$ is a modified query $q^\dagger(x)$.
	Then $\evaluate(k)$ is supposed to return what $q(x)$ would return in $\ADT$ in the corresponding sequence of calls.
	Formally, let $\hat u_{1\cdots k}$ be the sequence of update calls not of the form $q^\dagger(x)$ from the first $k$ elements of $\hat u^\dagger$, so $\hat u_{1\cdots k}$ is a sequence of calls to update operations of $\ADT$.
	Then $\lazy{\ADT}(\hat u^\dagger, \evaluate(k)) = \ADT(\hat u_{1\cdots k}, q(x))$, provided this is defined.
	If the $k$th call $u^\dagger_k$ is not of the form $q^\dagger(x)$, then the result of $\evaluate(k)$ is undefined.
\end{definitionx}

\begin{lemma}
	\label{lem:lazy transformation}
	Suppose a list of $m$ items can be maintained subject to looking up the $i$th item, and appending or removing a new item at the end, in $O(\ell(m))$ time per operation. (For word-RAM machines, $\ell(m)=1$; for pointer machines, $\ell(m)=\log m$.)

	Suppose $\ADT$ has an implementation in which each operation $f$ takes $O(T_f)$ time.
  Then $\lazy{\ADT}$ has an implementation in which $\evaluate$ takes $O(\ell(m))$ time, and each other operation $f$ or $f^\dagger$ takes $O(T_f+\ell(m))$ time, where $m$ is the total number of calls to operations.
\end{lemma}

\begin{proof}
To achieve $\lazy{\ADT}$ with the desired running times, we create an auxiliary list to store results of queries.
For a modified query call $q^\dagger(x)$ occurring as the $i$th update call, we store the corresponding result of $q(x)$ in position~$i$.
Then $\evaluate(i)$ can simply look up the value stored at position~$i$.
The $O(\ell(m))$ overhead comes from maintaining and accessing this list.
\end{proof}

Composing Lemmas~\ref{lem:rollback} and~\ref{lem:lazy transformation}, we obtain an implementation of $\partially{\lazy{\ADT}}$ in which $\evaluate$ takes time $O(\ell(m))$, and each other operation takes $O(m(T+\ell(m))$ time. We now show how a partially retroactive version of $\lazy{\ADT}$ is able to efficiently simulate a fully retroactive version of $\ADT$. The key idea is that the lazy version of a data structure has effectively converted the queries into update operations which are then allowed to be inserted at different times in the partially retroactive model.

\begin{theorem}
	\label{thm:lazy conversion}
	If $\partially{\lazy{\ADT}}$ has an implementation in which each operation takes amortized time $O(T)$, then so does $\fully{\ADT}$.
\end{theorem}

\begin{proof}
	We will use the implementation $\DS$ of $\partially{\lazy{\ADT}}$ to implement $\Insert{u}$, $\Delete$, and $\Query{q}$ with constant overhead. 
	The operations available to $\DS$ are $\Insert{u}$, $\Insert{q^\dagger}$, \Delete, and \evaluate. The collision of notation will not be a problem because the functions that share a name are handled by calling their namesake.

	Calls to $\Insert{u}$ and $\Delete$ can simply be run on $\DS$. On $\Query{q}(k,x)$,
	\begin{enumerate}
		\item Call $\Insert{q^\dagger}(k+1,x)$.
		\item Call $\evaluate(k+1)$.
		\item Call $\Delete(k+1)$.
		\item Output the result from step 2.
	\end{enumerate}
	This new data structure uses $\DS$ to maintain a list $\hat u^\dagger$ of calls to update operations of $\lazy{\ADT}$, but we make sure to immediately remove any calls to $q^\dagger$. Thus it equivalently maintains the list $\hat u$ of update calls to $\ADT$.

	When we call $\Query{q}(k,x)$, we insert the operation $q^\dagger(x)$ into $\hat u^\dagger$ at position $k+1$, and then use $\evaluate$ to extract the result of the corresponding $q(x)$ in $\hat u$. More formally, the result is
	$\lazy{\ADT}(\hat u^\dagger, \evaluate(k+1))$ (where $\hat u^\dagger$ is taken in the middle of the call, so it includes the $q^\dagger$), which by the definition of $\lazyop$ is $\ADT(\hat u_{1\cdots k}, q(x))$, as desired.
\end{proof}

This result allows us to transform a separation between $\ADT$ and $\fully{ADT}$ into a separation between $\lazy{\ADT}$ and $\partially{\lazy{\ADT}}$.
In particular, suppose $\ADT$ has an $O(f)$ implementation but $\fully{\ADT}$ requires $\Omega(g)$. Then by Lemma~\ref{lem:lazy transformation}, $\lazy{\ADT}$ is $O(f+\ell(m))$, and by the contrapositive of Theorem~\ref{thm:lazy conversion}, $\partially{\lazy{\ADT}}$ requires $\Omega(g)$. That is, a slowdown for fully retroactive $\ADT$ implies a slowdown for partially retroactive $\lazy{\ADT}$. If $\ell(m)=1$ such as for the word-RAM model, this slowdown is as large as the original.

Assuming any one of Conjectures~\ref{conj:circuitsat}--\ref{conj:3sum},
Chen et al.~\cite{chen2018nearly} define an ADT $\ADT$
with an $O(f)$ implementation and prove a lower bound of
$O(m^{1/2-o(1)} f)$ on implementing $\fully{\ADT}$.
Assuming $\ell(m)=m^{o(1)}$, the generic result of Theorem~\ref{thm:lazy conversion} implies that $\lazy{\ADT}$ has an $O(f)$-time implementation but $\partially{\lazy{\ADT}}$ requires $\Omega(m^{1/2-o(1)} f)$, giving an $\Omega(m^{1/2-o(1)})$ separation between an ADT and its partially retroactive version under each of the three assumptions.
In Sections \ref{sec:circuitsat}, \ref{sec:online min + product}, and \ref{sec:3sum}, we improve the separation to $m^{1-o(1)}$ when assuming Conjecture~\ref{conj:circuitsat}, \ref{conj:online min + product}, and~\ref{conj:3sum} respectively.

\section{\boldmath Lower Bound from SIZE$(2^{o(n)})$ CircuitSAT}\label{sec:circuitsat}

In this section, we prove a conditional gap for partial retroactivity using the ADT \defn{Circuit Counter}, which has the following operations and behavior:
\begin{itemize}
	\item \op{initialize}$(C)$: given a description of a circuit $C$ of size $r=2^{o(n)}$ which takes $n$ inputs, remember $C$. This can only be called once, and must be the first operation (otherwise the behavior is undefined). We assume $r>n$ for convenience.
	\item \op{set}$(x)$: given an $n$-bit string $x$, set the \defn{current string} to $x$.
	\item \op{increment}$()$: increment the current string as a binary number.
	\item \op{query}$()$: output True if \emph{any} past value of the current string satisfies $C$, and False otherwise.
\end{itemize}

\begin{lemma}\label{lem:circuitsat easy}
	There is an implementation of Circuit Counter in which each operation takes $2^{o(n)}$ time.
\end{lemma}

\begin{proof}
	The implementation will maintain what the current response to \op{query} would be; initially this is False. The rest is straightforward:
	\begin{itemize}
		\item \op{initialize}$(C)$ simply records $C$, taking time $O(r)=2^{o(n)}$.
		\item \op{set}$(x)$ sets the current string to $x$, and evaluates $C$ on it. If $C$ is satisfied, set the response to True. The runtime is dominated by evaluating $C$, which takes time $O(r)$.
		\item \op{increment}$()$ increments the current string, evaluates $C$ on it, and sets the response to True if $C$ is satisfied. This also takes time $O(r)$.
		\item \op{query}$()$ returns the current response, in time $O(1)$.
    \qedhere
	\end{itemize}
\end{proof}

\begin{theorem}\label{thm:circuitsat hard}
  There is a sequence of $\Theta(2^{n/2})$ operation calls
  for $\partially{\text{Circuit Counter}}$ such that,
	assuming Conjecture~\ref{conj:circuitsat},
  any implementation requires $2^{n-o(n)}$ total time,
  for an amortized lower bound of $n^{n/2-o(n)}$ time per operation.
\end{theorem}

\begin{proof}
	We will use a partially retroactive Circuit Counter to solve SIZE$(2^{o(n)})$ CircuitSAT. Given a circuit $C$ of size $r=2^{o(n)}$,
	we run the following sequence of operations:
	\begin{itemize}
		\item $\Insert{\op{initialize}}(1,C)$
		\item $2^{n/2}$ copies of $\Insert{\op{increment}}(2)$
		\item For each binary string $y$ of length $n/2$:
		\begin{itemize}
			\item $\Insert{\op{set}}(2,y0^{n/2})$
			\item \op{query}$()$
			\item \Delete$(2)$
		\end{itemize}
	\end{itemize}
	After the first three steps, the sequence of calls to the Circuit Counter is \op{initialize}$(C)$ and then $2^{n/2}$ copies of \op{increment}$()$. For each $y$, we insert \op{set}$(y0^{n/2})$ (i.e. \(y\) followed by \(n/2\) zero bits) right after the \op{initialize}, and query at the end. Since the \op{increment}s count through all strings starting with $y$, the \op{query} returns True if and only if any such string satisfies $C$. We then clean up the operation inserted before the next loop.

	After all $2^{n/2}$ loops, we have tested every possible input to $C$: thus $C$ is satisfiable if and only if any call to \op{query} returned True.
	Conjecture~\ref{conj:circuitsat} says that this whole algorithm must take $2^{n-o(n)}$ time.
\end{proof}

The implementation from Lemma~\ref{lem:circuitsat easy} uses $O(2^{o(n)})$ time per operation and we have $m=\Theta(2^{n/2})$ operations, so the gap here is $2^{n/2-o(n)}/2^{o(n)}=2^{n/2-o(n)}=m^{1-o(1)}$.

\section{\boldmath Lower Bound from Online $(\min,+)$ Product}\label{sec:online min + product}

In this section, we prove a conditional gap for partial retroactivity using the ADT \defn{$(\min,+)$ Multiplier}, which has the following operations and behavior:
\begin{itemize}
	\item \op{set-a}$(v)$: given an $n$-component row vector $v$, save it internally as $a$.
	\item \op{set-b}$(v)$: given an $n$-component column vector $v$, save it internally as $b$.
	\item \op{query}$()$: output the list of $a \diamond b$ for all historical values of the pair $(a,b)$.
  Here $a\diamond b$ is the $(\min,+)$ product $\min_i (a_i+b_i)$.
\end{itemize}

\begin{lemma}\label{lem:min + multiplier easy}
	There is an implementation of $(\min,+)$ Multiplier in which \op{set-a} and \op{set-b} take $O(n)$ time, and \op{query} takes $O(m)$ time, where $m$ is the number of calls to update operations.
\end{lemma}

\begin{proof}
	This is straightforward; we maintain the list of what \op{query} should output:
	\begin{itemize}
		\item \op{set-a}$(v)$ records $v$ as $a$, computes $a\diamond b$, and appends it to the list. This takes time $O(n)$.
		\item \op{set-b}$(v)$ records $v$ as $b$, computes $a\diamond b$, and appends it to the list. This also takes time $O(n)$.
		\item \op{query}$()$ returns the list, which takes time $O(m)$ because it has length $O(m)$.
    \qedhere
	\end{itemize}
\end{proof}

\begin{theorem}\label{thm:min + multiplier hard}
  There is a sequence of $\Theta(n)$ operation calls
  for $\partially{\text{$(\min,+)$ Multiplier}}$ such that,
	assuming Conjecture~\ref{conj:online min + product},
  any implementation requires $n^{3-o(1)}$ total time,
  for an amortized lower bound of $n^{2-o(1)}$ time per operation.
\end{theorem}

\begin{proof}
	We will use a partially retroactive $(\min,+)$ Multiplier to solve Online $(\min,+)$ Product. Given the $n \times n$ matrix $A$ and the $n$ vectors $v_i$ one by one,
	we run the following sequence of operations:
	\begin{itemize}
		\item For each row $A_j$ of $A$:
		\begin{itemize}
			\item $\Insert{\op{set-a}}(j,A_j)$
		\end{itemize}
		\item For each $v_i$:
		\begin{itemize}
			\item $\Insert{\op{set-b}}(1,v_i)$
			\item \op{query}$()$, and output the result
			\item \Delete$(1)$
		\end{itemize}
	\end{itemize}
	The first loop just sets up a sequence of $n$ \op{set-a} operations for the rows of $A$.
	The second loop inserts \op{set-b}$(v_i)$ at the beginning. Now the historical values of $(a,b)$ are precisely $(A_j,v_i)$ for each row $A_j$, so the result of the \op{query} is $A\diamond v_i$. Finally it removes the \op{set-b} to prepare for the next iteration.

	Conjecture~\ref{conj:online min + product} says that this whole algorithm must take $n^{3-o(1)}$ time.
\end{proof}

The sequence of operations in Theorem~\ref{thm:min + multiplier hard} has $m=\Theta(n)$ operations, in which case all operation running times from Lemma~\ref{lem:min + multiplier easy} are $O(n)$. So the gap is $n^{1-o(1)}=m^{1-o(1)}$.

\section{Lower Bound from 3SUM}\label{sec:3sum}

In this section, we prove a conditional gap for partial retroactivity using the ADT \defn{3-Summer}, which has the following operations and behavior:
\begin{itemize}
	\item \op{set-a}$(L)$: given a length-$\sqrt n$ list $L$, save it internally as $A$.
	\item \op{set-b}$(L)$: given a length-$\sqrt n$ list $L$, save it internally as $B$.
	\item \op{set-c}$(L)$: given a length-$n$ list $L$, save it internally as $C$.
	\item \op{query}$()$: return True if at any time, there has been a triple $(a,b,c)\in A\times B\times C$ satisfying $a+b+c=0$. We call such a triple \defn{good}.
\end{itemize}

Note that the internal lists $A$ and $B$ have length $\sqrt n$ while $C$ has length $n$.

\begin{lemma}\label{lem:3summer easy}
	There is an implementation of 3-Summer supporting each \op{set} operation in $O(n)$ time and \op{query} in $O(1)$ time.
\end{lemma}

\begin{proof}
	This is a bit more complicated than the implementations from the previous two sections. As usual, we maintain the value that \op{query} should return.
	The key fact is that we can test whether there is \emph{currently} a good triple in time $O(n)$: make a hash table containing the values of $C$, then iterate through all $n$ values of $(a,b)\in A\times B$, and check whether $-a-b$ is in the hash table. The rest is straightforward:
	\begin{itemize}
		\item \op{set-a}$(L)$ records $L$ as $A$, checks if there's a good triple, and sets the query value to True if so. The runtime is dominated by checking for a good triple, which takes time $O(n)$.
		\item \op{set-b}$(L)$ records $L$ as $B$, and then proceeds just like \op{set-a}.
		\item \op{set-c}$(L)$ records $L$ as $C$, and checks for a good triple. Both steps take $O(n)$ time.
		\item \op{query}$()$ returns the current query value, in time $o(1)$.
    \qedhere
	\end{itemize}
\end{proof}

\begin{theorem}\label{thm:3summer hard}
  There is a sequence of $\Theta(\sqrt n)$ operation calls
  for $\partially{\text{3-Summer}}$ such that,
	assuming Conjecture~\ref{conj:3sum},
  any implementation requires $n^{2-o(1)}$ total time,
  for an amortized lower bound of $n^{3/2-o(1)}$ time per operation.
\end{theorem}

\begin{proof}
	We will use a partially retroactive 3-Summer to solve 3SUM.
	Given the lists $A$, $B$, and $C$ of length $n$, we first divide $A$ into $\sqrt n$ lists $A_1,\dots,A_{\sqrt n}$ of length $\sqrt n$, and similarly for $B$.%
	\footnote{For simplicity we're assuming $n$ is a perfect square; otherwise one can take ceilings as appropriate.}
	Now we run the following sequence of operations on the partially retroactive 3-Summer:
	\begin{itemize}
		\item $\Insert{\op{set-c}}(C)$
		\item For $j$ from $1$ to $\sqrt{n}$:
		\begin{itemize}
			\item $\Insert{\op{set-b}}(j+1,B_j)$
		\end{itemize}
		\item For $i$ from $1$ to $\sqrt{n}$:
		\begin{itemize}
			\item $\Insert{\op{set-a}}(2,A_i)$
			\item \op{query}$()$
			\item \Delete$(2)$
		\end{itemize}
	\end{itemize}
	The first two steps set up operations to initialize $C$, and then set the internal list for $B$ to each section of $B$ in order.
	Then, for each section $A_i$ of $A$, we insert \op{set-a}$(A_i)$ right after the \op{set-c}. Now the historical values of $(A,B)$ are $(A_i,B_j)$ varying over $j$, so the \op{query} returns True if and only if there is a good triple $(a,b,c)\in A_i\times B\times C$. Finally, we remove the \op{set-a} to prepare for the next iteration.

	Through the second loop, we check each section of $A$ for a good pair. Ultimately, at least one of the calls to \op{query} returns True if and only if there is a good pair $(a,b,c)\in A\times B\times C$.
	Conjecture~\ref{conj:3sum} says that this whole algorithm must take $n^{2-o(1)}$ time.
\end{proof}

The gap between the retroactive lower bound in Theorem~\ref{thm:3summer hard} and the nonretroactive upper bound in Lemma~\ref{lem:3summer easy} is $n^{1/2-o(1)}$.
Because the sequence of operations in Theorem~\ref{thm:3summer hard} uses $m=\Theta(\sqrt n)$ operations, this gap is again $m^{1-o(1)}$.

\section{Conclusion}
\label{sec:Conclusion}

This paper gives compelling evidence that there is no general way to make certain data structures retroactive with significantly better running time than the naive rollback method.
This worst-case impossibility motivates the question of which problems
\emph{can} be made retroactive with low overhead
(say, a polylogarithmic factor),
such as deques \cite{demaine2007retroactive},
priority queues \cite{demaine2007retroactive,FullyRetroactive_WADS2015}, predecessor \cite{retroactive-predecessor}, and data structures whose updates all commute\cite{demaine2007retroactive}.

On the other side, the $L$-evaluation framework of Chen et al.~\cite{chen2018nearly} appears to be a fairly general method for taking a problem with a computational lower bound and using it to generate a gap between partial and full retroactivity. Perhaps the methods in that paper and this one can be generalized to understand in a broader sense what makes data structures resistant to being made retroactive. This perspective may also be useful in the other direction, granting some intuition about what makes algorithmic problems computationally difficult.

Our applications to the computational cost of time travel fail to account
for the physical world's ability to perform quantum computation.
It would be interesting to extend our results to obtain similar separations
for quantum models of computation, where in particular search can be
sped up by Grover's algorithm.
See \cite{Buhrman-Patro-Speelman-2021} for quantum analogs to
some of our fine-grained assumptions.

\bibliography{main.bib}

\end{document}